\documentclass[submission,copyright,creativecommons]{eptcs}
\providecommand{\event}{SOS 2007} % Name of the event you are submitting to
\usepackage{breakurl}             % Not needed if you use pdflatex only.
\usepackage{underscore}           % Only needed if you use pdflatex.

\usepackage{amsthm}
\usepackage{amssymb}
\usepackage{amsmath}
\usepackage{amsfonts}

\usepackage{mdwtab}
\usepackage{latexsym}
\usepackage{ifthen}
\usepackage[utf8]{inputenc}

\newlength{\fitchlineht}
\newlength{\fitchindent}
\newlength{\fitchcomind}
\newlength{\fitchnumwd}
\makeatletter
\newcommand\fvline[1][\arrayrulewidth]{\vrule\@height.5\fitchlineht\@width#1\relax}
\makeatother
\newcommand{\fa}{\vline\hspace*{\fitchindent}}
\newcommand{\fh}{\fvline%
  \makebox[0pt][l]{{%
      \raisebox{-1.4ex}[0pt][0pt]{\rule{1.5em}{\arrayrulewidth}}}}%
  \hspace*{\fitchindent}}
\newcounter{fitchcounter}
\newboolean{resetfitchcounter}
\setboolean{resetfitchcounter}{true}
\newboolean{increasefitchcounter}
\setboolean{increasefitchcounter}{true}
\newcommand{\formatfitchcounter}[1]{\arabic{#1}}
\newcommand{\fitchcounter}{%
  \ifthenelse{\boolean{increasefitchcounter}}{\addtocounter{fitchcounter}{1}}{}
  \formatfitchcounter{fitchcounter}}

\newenvironment{fitchnum}%
{\ifthenelse{\boolean{resetfitchcounter}}{\setcounter{fitchcounter}{0}}{}
  \begin{tabular}{!{\makebox[\fitchnumwd][r]{\fitchcounter }\hspace{\fitchindent}}Ml@{\hspace{\fitchcomind}}l}}%
{\end{tabular}}

\newenvironment{fitchunum}%
{\begin{tabular}{!{\makebox[\fitchnumwd][r]{}\hspace{\fitchindent}}Ml@{\hspace{\fitchcomind}}l}}%
{\end{tabular}}

\newenvironment{fitch}{
  \begin{fitchnum}}{\end{fitchnum}}
\newenvironment{fitch*}{
  \begin{fitchunum}}{\end{fitchunum}}

{\begin{eqnarray}
    &#1\label{#2}\\
    &\begin{fitch}}%
    {\end{fitch}\notag\end{eqnarray}}

\newcommand{\IMPL}{\rightarrow}

\newcommand{\AND}{\wedge}

\newcommand{\OR}{\vee}

\newcommand{\NOT}{\neg}

\newcommand{\POSS}{\lozenge}

\newtheorem{theorem}{Theorem}
\newtheorem{lemma}{Lemma}

\newtheorem{proposition}{Proposition}

\title{Natural Deduction for Assertibility and Deniability\thanks{This paper is an outcome of the project Logical Structure of Information Channels, no. 21-23610M, supported by the Czech Science Foundation and realized at the Institute of Philosophy of the Czech Academy of Sciences.}}
\author{V\'it Pun\v coch\'a\v r
\institute{Institute of Philosophy, \\Czech Academy of Sciences, \\Czech Republic}
\email{puncochar@flu.cas.cz}
\and
Berta Grimau
\institute{Institute of Philosophy, \\Czech Academy of Sciences, \\Czech Republic}
\email{grimau@utia.cas.cz}
}
\def\titlerunning{Natural Deduction for Assertibility and Deniability}
\def\authorrunning{V. Pun\v coch\'a\v r, B. Grimau}
\begin{document}
\maketitle

\begin{abstract}
In this paper we split every basic propositional connective into two versions, one is called extensional and the other one intensional. The extensional connectives are semantically characterized by standard truth conditions that are relative to possible worlds.  Assertibility and deniability of sentences built out of atomic sentences by extensional connectives are defined in terms of the notion of truth.  The intensional connectives are characterized directly by assertibility and deniability conditions without the notion of truth. We pay special attention to the deniability condition for intensional implication. We characterize the logic of this mixed language by a system of natural deduction that sheds some light on the inferential behaviour of these two kinds of connectives and on the way they can interact. 
\end{abstract}

\section{Introduction}

Christopher Gauker, in his book \cite{gauker05}, put forward an interesting theory of conditionals based on the notions of assertibility and deniability in a context. A simplification of Gauker's theory was proposed in  \cite{puncochar16} where it was shown that the characteristic features of the logic determined by Gauker's semantics are preserved even if we replace Gauker's  rather complex notion of context with a much simpler one, according to which contexts are represented by sets of possible worlds. Besides simplicity the proposed modification has some further nice properties that the original Gauker's theory lacks, concerning a simple form of compositionality, validity of some plausible argument forms, or simple treatment of conditionals embedded in antecedents of other conditionals (for more details, see \cite{puncochar16}).

A peculiar feature of Gauker's theory is that it reflects one rather surprising phenomenon: Logical operators are sensitive to the types of sentences they are applied to. For example, disjunction of two conditional sentences seems to behave differently than disjunction of two elementary sentences. The original semantics from \cite{gauker05} incorporates this ambiguity directly into the formal language. This approach was further elaborated, especially from the philosophical point of view, in \cite{puncochargauker20}. The semantics from \cite{puncochar16} is designed in a different way. It disambiguates the behaviour of logical connectives at the level of the formal language by replacing each connective with two operators, one ``extensional'' and one ``intensional''. It is argued in \cite{puncochar16} that  while the original Gauker's approach allows us to model some phenomena  in a straightforward way, it is technically less elegant than the disambiguation strategy. 

In this paper, we will focus on the semantics of extensional and intensional connectives proposed in \cite{puncochar16} and we will study the logic determined by this framework. Let us call it \textit{the Logic of Assertibility and Deniability}, or \textsf{LAD}, for short. \cite{puncochar16} provided a syntactic characterization of \textsf{LAD}, but only an indirect one, via a translation into a modal logic. The main contribution of this paper is a direct syntactic characterization of \textsf{LAD}. We will show that this logic can be characterized by an elegant system of natural deduction that clarifies the inferential behaviour of both extensional and intensional operators.

%This proposal was further elaborated, especially from the philosophical point of view, in \cite{puncochargauker20}.

%The semantic framework used in  \cite{puncochargauker20} deviates from the proposal of \cite{puncochar16} in two respects. First, the semantics from \cite{puncochar16} disambiguates at the level of the formal language the above mentioned sensitivity of logical operators by replacing each logical connective with two operators, one ``extensional'' and one ``intensional''. In contrast, the semantics from \cite{puncochargauker20} (as well as from \cite{gauker05}) incorporates this ambiguity directly into the formal language. This theoretical decision allows us to model some phenomena  in a straightforward way but it is technically less elegant than the disambiguation strategy (as argued in \cite{puncochar16}). Second, while both the original theory from \cite{gauker05} and its proposed modification from \cite{puncochar16} characterize negation in terms of deniability conditions, \cite{puncochargauker20} treats negation in quite a different way, as assertibility of falsity when negation is applied to conditional-free sentences, and as falsity of assertibility if it is applied to sentences involving conditionals.

%\cite{puncochar14}, \cite{puncochar15}

\section{Extensional and intensional connectives}\label{extint}

Consider the following scenario. A murder has been committed and we have four suspects. The murderer is not known but it is settled that it must be someone among these suspects. It is also clear that only one person has committed the crime. We have the following description of the suspects:

\vspace{6pt}

\begin{tabular}{ll}
tall man with dark hair and moustache & tall man with blond hair and without moustache \\
short man with blond hair and moustache & short man with dark hair and without moustache \\
\end{tabular}

\vspace{6pt}

Now it seems that in this context one is entitled to assert the premises but not the conclusion of the following argument that has a seemingly valid form:

\vspace{6pt}

\begin{tabular}{ll}
Premise 1 &\textit{The murderer is either tall or short.} \\
Premise 2 &\textit{If the murderer is tall then he has a moustache if he has dark hair.} \\
Premise 3 &\textit{If the murderer is short then  he has a moustache if he has blond hair.} \\
\hline Conclusion &\textit{It either holds  that the murderer has a moustache if he has dark hair} \\
& \textit{or that the murderer has a moustache if he has blond hair.} \\
\end{tabular}

\vspace{6pt}

The assertibility of the premisses is clear from the description of the context. But the conclusion does not seem to be assertible. As a reason for this intuition one might say that none of the two disjuncts in the conclusion  is assertible in the context. But the same holds for the first premise, none of the disjuncts is assertible in the context and yet the whole disjunction clearly is. This puzzling phenomenon illustrates what we vaguely described in the introduction as the sensitivity of logical operators. The first premise says that in each possibility (``possible world''), one of the disjuncts is true. The disjunction connecting two conditionals in the conclusion says something different, namely that at least one of the two conditionals holds with respect to the whole context. We can describe these two cases as involving two different logical operators, one operating on the level of possible worlds and the other operating on the level of the whole context. One can easily formulate similar examples that involve negation of conditionals (see \cite{puncochar16}). For a discussion of similar phenomena, see, e.g. \cite{yalcin12} or \cite{bledin14}. 

Such examples motivate splitting each of the basic logical connectives (implication, conjunction, disjunction, and negation) into two versions, one will be called ``extensional'' and the other one ``intensional'': 

\vspace{6pt}

\begin{tabular}{lcccc}
Extensional connectives: & $\supset$ & $\cap$ & $\cup$ & $\sim$ \\
Intensional connectives: & $\IMPL$ & $\AND$ & $\OR$ & $\NOT$ \\
\end{tabular}

\vspace{6pt}

%As motivated by the example, the idea is that extensional connectives are connectives that can be characterized by truth conditions relative to possible worlds. Intensional connectives cannot be characterized in this way. Instead they are characterized in terms of assertibility and deniability conditions that are relative to contexts and contexts are entities of different kind than possible worlds. This is elaborated below in the precise definition of the semantics.

Let $L$ be the language containing all atomic formulas and all formulas which can be constructed out of the atomic formulas using the extensional connectives. The language $L^*$ is the smallest set of formulas containing all $L$-formulas and closed under the application of the intensional connectives. The two languages can be concisely introduced as follows:

\vspace{6pt}

\begin{tabular}{lc}
$L$: & $\alpha= p \mid \alpha \supset \alpha \mid \alpha \cap \alpha \mid \alpha \cup \alpha \mid {\sim}\alpha$\\
$L^*$: & $\varphi= \alpha \mid \varphi \IMPL \varphi \mid \varphi \AND \varphi \mid \varphi \OR \varphi \mid \NOT \varphi$
\end{tabular}

\vspace{6pt}

The argument above would be formalized in the language $L^*$ in the following way (note that the language forbids us from connecting two intensional implications by the extensional disjunction):

\vspace{6pt}

\begin{tabular}{l}
 $p \cup q, p \IMPL (r \IMPL t), q \IMPL (s \IMPL t) / (r \IMPL t) \OR (s \IMPL t)$
\end{tabular}

%\begin{tabular}{ll}
%Premise 1 & $p \cup q$ \\
%Premise 2 & $p \IMPL (r \IMPL t)$ \\
%Premise 3 & $q \IMPL (s \IMPL t)$ \\
%\hline Conclusion & $(r \IMPL t) \OR (s \IMPL t)$ \\
%\end{tabular}

\vspace{6pt}

The Greek letters $\alpha, \beta, \gamma$ will range over $L$-formulas and the letters $\varphi, \psi, \chi, \vartheta$ over $L^*$-formulas. %If $A$ is a particular set of atomic formulas then $L^*_A$-formulas (or $L_A$-formulas) are those $L^*$-formulas (or $L$-formulas) that contain no other atomic formulas than those from $A$.
Let $A$ be a set of atomic formulas. A possible $A$-world is any function that assigns a truth value ($1$~representing \textit{truth} or $0$ representing \textit{falsity}) to every atomic formula from $A$. An $A$-context is a nonempty set of possible $A$-worlds. The specification of the set $A$ will be omitted if no confusion arises. Note that, by definition, there cannot be two different possible $A$-worlds in which exactly the same atomic formulas are true.

For the complex formulas of the language $L$, the truth and falsity conditions with respect to individual possible worlds are those of classical propositional logic. $\Vdash^{+}$ and $\Vdash^{-}$ will now respectively stand for the relations of assertibility and deniability between contexts and formulas of the language $L^*$. The assertibility and deniability conditions taken from \cite{puncochar16} (and motivated by \cite{gauker05})   are defined in the following way:

\vspace{6pt}

\begin{tabular}{l}
$C \Vdash^{+} \alpha$ iff for all $w \in C$, $\alpha$ is true in $w$.\\
$C \Vdash^{-} \alpha$ iff for all $w \in C$, $\alpha$ is false in $w$.\\
$C \Vdash^{+} \NOT\varphi$ iff $C \Vdash^{-} \varphi$.\\
$C \Vdash^{-} \NOT \varphi$ iff $C \Vdash^{+} \varphi$.\\
$C \Vdash^{+} \varphi \OR \psi$ iff $C \Vdash^{+} \varphi$ or $C \Vdash^{+} \psi$.\\
$C \Vdash^{-} \varphi \OR \psi$ iff $C \Vdash^{-} \varphi$ and $C \Vdash^{-} \psi$.\\
$C \Vdash^{+} \varphi \AND \psi$ iff $C \Vdash^{+} \varphi$ and $C \Vdash^{+} \psi$.\\
$C \Vdash^{-} \varphi \AND \psi$ iff $C \Vdash^{-} \varphi$ or $C \Vdash^{-} \psi$.\\
$C \Vdash^{+} \varphi \rightarrow \psi$ iff $D \Vdash^{+} \psi$ for all nonempty $D \subseteq C$, such that $D \Vdash^{+} \varphi$.\\
$C \Vdash^{-} \varphi \rightarrow \psi$ iff $D \Vdash^{-} \psi$ for some nonempty $D \subseteq C$, such that $D \Vdash^{+} \varphi$.\\
\end{tabular}

\vspace{6pt}

The consequence relation is defined as assertibility preservation. That is, a set of $L^*$-formulas $\Delta$ entails an $L^*$-formula $\psi$ (symbolically $\Delta \vDash \psi$) if $\psi$  is assertible in every context in which all formulas from $\Delta$ are assertible. This consequence relation represents what we here call  \textit{the Logic of Assertibility and Deniability} (\textsf{LAD}). We say that two $L^*$-formulas, $\varphi$  and $\psi$, are (logically) equivalent in this logic (symbolically $\varphi \equiv \psi$) if they are assertible in the same contexts, that is, if $\varphi \vDash \psi$  and $\psi \vDash \varphi$. They are strongly equivalent (symbolically $\varphi \rightleftharpoons \psi$) if they are not only assertible but also deniable in the same contexts, that is, if $\varphi \equiv \psi$ and $\NOT \varphi \equiv  \NOT \psi$. It can be shown by induction that strongly equivalent formulas are universally interchangeable (this does not hold for mere equivalence). Note that the Gauker's original semantics from \cite{gauker05} does not have this property.
\begin{proposition}
Assume that $\varphi \rightleftharpoons \psi$, and $\varphi$  occurs as a subformula in $\chi$. If $\chi'$ is the result of replacing an occurrence of the subformula $\varphi$  in $\chi$ with $\psi$ then $\chi \rightleftharpoons \chi'$.
\end{proposition}
Note that the assertibility clause for disjunction formally corresponds to the support condition for inquisitive disjunction used in inquisitive semantics (see \cite{ciardelliroelofsen11}). However, in contrast to inquisitive semantics we do not treat intensional disjunction as a question generating operator but as a statement generating operator (e.g., the formula $(r \IMPL t) \OR (s \IMPL t)$ in the conclusion of our example does not represent a question but a statement). The assertibility clauses for conjunction and implication also correspond to the support clauses for these connectives in inquisitive semantics. Technically speaking, if negation is omitted, \textsf{LAD} just corresponds to inquisitive logic (which is axiomatized for example in \cite{ciardelliroelofsen11}). What is different from the standard inquisitive logic is our treatment of negation via the deniability conditions. The deniability conditions for (intensional) negation, disjunction and conjunction are the standard ones, typically used in a bilateralist setting (see, e.g., \cite{odintsov13}).  The most tricky case is the deniability condition for implication. There are two natural candidates for an alternative deniability clause:

\vspace{6pt}

\begin{tabular}{ll}

(A) & $C \Vdash^{-} \varphi \rightarrow \psi$ iff $C \Vdash^{+} \varphi$ and $C \Vdash^{-} \psi$.\\
(B) & $C \Vdash^{-} \varphi \rightarrow \psi$ iff $D \Vdash^{-} \psi$ for all nonempty $D \subseteq C$, such that $D \Vdash^{+} \varphi$.\\
\end{tabular}

\vspace{6pt}

The option (A) is known from Nelson logic (see \cite{odintsov13}). However, this option licenses the inference from $\NOT (\varphi \IMPL \psi)$ to $\varphi$ (and also to $\NOT \psi$), which we found highly problematic from the natural language point of view (consider the following argument: \textit{It is not the case that if I die today, I will be living tomorrow. Therefore I will die today.}). The option (B) looks much more plausible. This option is known from connexive logic (see \cite{wansing21}). It makes  $\NOT (\varphi \IMPL \psi)$ strongly equivalent to $\varphi \IMPL \NOT \psi$. This is indeed a reasonable way to read negations of conditionals. This clause was explored in more detail within the semantics of assertibility and deniability in \cite{puncochar14}. 

In our current setting the clause (B) would allow us to completely eliminate intensional negations from any $L^*$-formula. Since we also have double negation and DeMorgan laws guaranteed by the other semantic clauses, we could push all the intensional negations occurring in a formula to the subformulas that are already in the language $L$. For example, in the formula $\NOT (\NOT (p \AND {\sim}s) \IMPL (p \cup {\sim}q))$ we could push the negations inside the formula in the following way: $(\NOT p \OR \NOT{\sim}s) \IMPL \NOT (p \cup {\sim}q)$. Moreover, observe that for any $L$-formula $\alpha$ we have $\NOT \alpha \rightleftharpoons {\sim}\alpha$, so when pushed so that it applies to an $L$-formula, intensional negation can be replaced with extensional negation (we would obtain $({\sim}p \OR {\sim}{\sim}s) \IMPL {\sim}(p \cup {\sim}q)$, in our example). Hence, if the deniability clause (B) for implication is employed, every occurrence of intensional negation in any $L^*$-formula  can be dissolved and thus intensional negation does not add any extra expressive power.

Aside from making intensional negation redundant, the clause (B) has one other questionable consequence. If this clause is employed, there are formulas that are assertible as well as deniable in some contexts. There are even formulas that are assertible as well as deniable in every context and thus there are formulas of the form $\varphi \AND \NOT \varphi$ regarded as logically valid. A concrete example is obtained when we substitute $(p \AND \NOT p) \IMPL (p \OR \NOT  p)$ for $\varphi$. Hence, the resulting logic is not only paraconsistent but it is in this sense inconsistent (see \cite{wansing21} for a more detailed discussion of this phenomenon). From a certain point of view, this may not be conceived as a substantial problem. However, in this paper, we want to maintain consistency and thus we reject clause (B) in favour of the clause stated above according to which a conditional is deniable in a context $C$ if and only if there is some subcontext of $C$ in which the antecedent is assertible and the consequent is deniable. This deniability clause corresponds to the one proposed in Gauker's book \cite{gauker05}. It can be viewed as a weakening of the strong clause (A) in the sense that $\NOT (\varphi \IMPL \psi)$ is not equivalent to $\varphi \AND \NOT \psi$  but rather to a mere possibility of $\varphi \AND \NOT \psi$. This way of denying conditionals seems to have strong support from natural language. For example, with respect to the context above, one can deny the claim \textit{if the murderer is tall, he has dark hair} because it is possible (from the perspective of the context) that the murderer is tall and does not have dark hair. In comparison to (B) we can observe that the clause (C) allows us to prove by induction the following fact. 
\begin{proposition}
There is no context $C$ and no $L^*$-formula $\varphi$  such that $C \Vdash^+ \varphi$ and $C \Vdash^- \varphi$.
\end{proposition}
Another feature that distinguishes our deniability clause for implication from (B) is that it leads to the classical behaviour of intensional connectives in the special case where the context contains only one world. Thus in singleton contexts the assertibility and deniability clauses for intensional connectives coincide with truth and falsity clauses for the extensional connectives. To express this fact more formally, we will use the following notation. For any $L^*$-formula $\varphi$, $\varphi^{e}$ denotes the $L$-formula which is the result of replacing all intensional connectives in $\varphi$ with their extensional counterparts. Then the following fact can be easily established by induction.
\begin{proposition}\label{p3}
$\left\{w\right\} \Vdash^{+} \varphi$ iff $\varphi^{e}$ is true in $w$, and $\left\{w\right\} \Vdash^{-} \varphi$ iff $\varphi^{e}$ is false in $w$.
\end{proposition}

Since the truth conditions for the extensional connectives are the standard ones, the consequence relation restricted to the language $L$ is classical.
\begin{proposition}
For the formulas of the language $L$ the consequence relation coincides with the consequence relation of classical logic.
\end{proposition}
Of course, in general intensional and extensional connectives behave differently. Let us illustrate this with the murder scenario described above. Let $A$  be the set of atomic formulas $\{p, q, r, s, t \}$. Consider the following formalization and an $A$-context consisting of four possible $A$-worlds that differ from each other on who among the suspects committed the murder. 

\vspace{6pt}

\begin{tabular}{llcccc}
&& $w_1$ & $w_2$ & $w_3$ & $w_4$ \\
\textit{The murderer is tall.} &$p$ & 1 & 1 & 0 & 0 \\
\textit{The murderer is short.} &$q$ & 0 & 0 & 1 & 1 \\
\textit{The murderer has dark hair.} &$r$ & 1 & 0 & 0 & 1\\
\textit{The murderer has blond hair.} &$s$ & 0 & 1 & 1 & 0 \\
\textit{The murderer has a moustache.} &$t$ & 1 & 0 & 1 & 0 \\
\end{tabular}

\vspace{6pt}

Now one can check that in this context all the premises of the argument, namely the formulas $p \cup q$, $p \IMPL (r \IMPL t)$, and $q \IMPL (s \IMPL t)$, are assertible but the conclusion $(r \IMPL t) \OR (s \IMPL t)$  is not.

Our deniability condition for conditionals is an important ingredient of the semantics because it increases the expressive power of the language and introduces a specific kind of non-persistent formulas. We say that a formula is \textit{persistent} if it holds that whenever it is assertible in a context, it is assertible in every subcontext (i.e. in every nonempty subset) of the context. It can be proved by induction that every $L^*$-formula that does not contain any intensional implication in the scope of an intensional negation is persistent.

This claim does not hold generally. For example, for the above context $C=\{ w_1, w_2, w_3, w_4 \}$ we have $C \Vdash^+ \NOT (p \IMPL q)$ because there is a subcontext $D=\{w_1, w_2 \}$ such that $D \Vdash^+ p$ and $D \Vdash^- q$. However, there is a subcontext $E \subseteq C$, e.g. $E = \{w_3, w_4 \}$, such that $E  \nVdash^+ \NOT (p \IMPL q)$.

Note that every formula of the form $\varphi \IMPL \psi$  is also persistent, even if it contains an occurrence of intensional implication in the scope of intensional negation. We would like to say that persistent formulas play  in inferences a substantially different role than non-persistent formulas.  However, the notion of a persistent formula was defined \textit{semantically} and thus we cannot use it directly in the formulation of a syntactic deductive system. Nevertheless, we just specified \textit{syntactically} a large class of formulas that have this semantic property. We will call them safe.  We say that a formula is \textit{safe}, if either it does not contain any $\IMPL$ in the scope of $\NOT$, or it is of the form $\varphi \IMPL \psi$. 
\begin{proposition}\label{l: persistence of safe formulas}
Every safe formula is persistent.
\end{proposition}
Safe formulas will play an important role in the formulation of our system of natural deduction. The subsequent completeness proof will indicate that this syntactic notion sufficiently approximates the notion of a persistent formula.

Now we will show that by the extra expressive power gained from the deniability clause for implication we obtain functional completeness. To describe this fact more precisely, we will use the following defined symbols that will also be  used later in the system of natural deduction. 

\vspace{6pt}

\begin{tabular}{ll}
(a) & $\bot =_{def} p \cap {\sim} p$, for some selected fixed atomic formula $p$,\\
%\item[2.] $\top =_{def} p \OR \NOT p$, (logical truth)
%\item[3.] $\NEC \varphi =_{def} \top \IMPL \varphi$, (first kind of necessity: $\varphi$ is assertible in every subcontext) 
(b) & $\POSS \varphi =_{def} \NOT (\varphi\IMPL \bot)$,\\
(c) & $\alpha_1 \oplus \ldots \oplus  \alpha_n =_{def} (\alpha_1 \cup \ldots \cup \alpha_n) \AND (\POSS \alpha_1 \AND \ldots \AND \POSS \alpha_n)$.
\end{tabular}

\vspace{6pt}

The symbol $\bot$ represents a \textit{contradiction}. It can be easily observed that $C \Vdash^- \bot$, for every context $C$. The symbol $\POSS$ expresses a contextual possibility. It holds that $C \Vdash^+ \POSS \varphi$ iff there is a nonempty $D \subseteq C$ such that $D \Vdash^+ \varphi$. Finally, $\oplus$ represents a ``pragmatic disjunction'', i.e. an extensional disjunction but such that all its disjuncts are possible (which is usually regarded to be a pragmatic feature of disjunction).

If a finite set of atomic formulas is fixed, any $A$-world $w$ can be described by an $L$-formula $\sigma_w$ in the usual way. For example, in the above example the formula $\sigma_{w_1}$ would be $p \cap {\sim}q \cap r \cap {\sim}s \cap t$. Now any $A$-context $C=\{w_1, \ldots, w_n \}$ can be described by the formula $\mu_C=\sigma_{w_1} \oplus \ldots \oplus \sigma_{w_n}$. Moreover, any set of $A$-contexts $X=\{ C_1, \ldots, C_k \}$ can be described by the formula $\xi_X=\mu_{C_1} \OR \ldots \OR \mu_{C_k}$. Now we can state the functional completeness result.

\begin{proposition}
Let $A$  be a finite set of atomic formulas. Then it holds for any $A$-contexts $C, D$ that $C \Vdash^+ \mu_D$ iff $D=C$. Moreover, it holds for any set of $A$-contexts $X$, and any $A$-context $C$ that $C \Vdash^+ \xi_X$ iff $C \in X$.
\end{proposition}
A similar result was observed in \cite{puncochar15} for inquisitive logic with an operation called ``weak negation''. In fact, the crucial part of the proof of our main result will be a syntactic reconstruction of weak negation within our system of natural deduction.

\section{Deductive calculus} 

In this section we define a Fitch style system of natural deduction for the semantics of assertibility and deniability. We will use the following economic notation. We use brackets and the colon to refer to subproofs.  For example, $[\varphi: \psi]$ stands for a subproof in which $\varphi$ is the hypothetical assumption and $\psi$ is the last line of the subproof. We use two kinds of brackets, square and round.  The difference is in the formulas from the outer proof that can be used in the subproof. By using square brackets in $[\varphi: \psi]$ we indicate that all formulas that are available in the step before the hypothetical assumption $\varphi$ is made are also available under this assumption in the derivation of $\psi$ from $\varphi$. This is just as in the standard natural deduction systems for classical and intuitionistic logic. The round brackets indicate that there is a restriction concerning the formulas from the outer proof that can be used under the assumption. By writing $(\varphi: \psi)$, we indicate that only \textit{safe} formulas from the outer proof can be used in the derivation of $\psi$ from $\varphi$. (Safe formulas are defined in the previous section.) 

The distinction between square and round brackets reflects the semantic distinction between two kinds of hypothetical assumption. Square brackets, i.e. $[\varphi: \psi]$, indicate that by making the hypothetical assumption $\varphi$ in a context $C$ we do not change the given context, we only assume that the assumption $\varphi$ is assertible in $C$. In contrast, round brackets, i.e. $(\varphi: \psi)$,  indicate  that by making the hypothetical assumption $\varphi$ in a context $C$ we move from $C$ to an arbitrary subcontext $D$ of $C$ in which $\varphi$ is assertible. If we have already proved that a formula is assertible in $C$, we can use it under the assumption $\varphi$ (in $D$) only if it is guaranteed that it is persistent. Proposition \ref{l: persistence of safe formulas} guarantees that all safe formulas are persistent.

We split the rules of the calculus into three groups. The first group contains the ``classical'' introduction and elimination rules for the extensional connectives.

\vspace{6pt}

\begin{tabular}{llll}
(i$\cap$) & $\alpha, \beta / \alpha \cap \beta$ & (e$\cap_1$) & $\alpha \cap \beta / \alpha$ \\
&& (e$\cap_2$) & $\alpha \cap \beta / \beta$ \\[0,3cm]
(i$\cup_1$) & $\alpha / \alpha \cup \beta$ & (e$\cup$) & $\alpha \cup \beta, (\alpha:\gamma), (\beta:\gamma) / \gamma$\\
(i$\cup_2$) & $\beta / \alpha \cup \beta$ && \\[0,3cm]
(i$\supset$) & $(\alpha:\beta)/\alpha \supset \beta$ & (e$\supset$) & $\alpha, \alpha \supset \beta / \beta$\\[0,3cm]
(i$\sim$) & $(\alpha:\bot)/{\sim}\alpha$ & (e$\sim_1$) & $\alpha, {\sim}\alpha/\bot$\\
&& (e$\sim_2$) & ${\sim}{\sim}\alpha/\alpha$\\
\end{tabular}

\vspace{6pt}

\noindent The second group contains the following ``intuitionistic''  rules concerning the intensional connectives (but notice the restrictions indicated by round brackets and the fact that (i$\NOT$) is restricted to $L$-formulas):

\vspace{6pt}

\begin{tabular}{llll}
(i$\AND$) & $\varphi,\psi/\varphi \AND \psi$ & (e$\AND_1$) & $\varphi \AND \psi/ \varphi$ \\
&& (e$\AND_2$) & $\varphi \AND \psi/ \psi$ \\[0,3cm]
(i$\OR_1$) & $\varphi/ \varphi \OR \psi$ & (e$\OR$) & $\varphi \OR \psi, [\varphi:\chi], [\psi:\chi] / \chi$\\
(i$\OR_2$) & $\psi/ \varphi \OR \psi$ && \\[0,3cm]
(i$\IMPL$) & $(\varphi:\psi)/\varphi \IMPL \psi$ & (e$\IMPL$) & $\varphi, \varphi \IMPL \psi/ \psi$\\[0,3cm]
(i$\NOT$) & $(\alpha:\bot)/\NOT \alpha$ & (e$\NOT$) & $\varphi, \NOT \varphi / \bot$\\
&& (EFQ) & $\bot /\varphi$
\end{tabular}

\vspace{6pt}

%Note that there is a restriction concerning (i$\IMPL$) but not (e$\OR$), which is indicated by the two kinds of brackets. 
%The third group consists of the rules that characterize the interaction of intensional negation with all intensional operators, plus three rules, ($\POSS\OR$), ($\POSS\oplus$), and (CEM), i.e. ``contextual excluded middle'':
The third group consists of the rules that characterize the interaction of intensional negation with all intensional operators, plus two extra rules, (CEM), i.e. ``contextual excluded middle'', and ($\POSS\oplus$):

\vspace{6pt}

\begin{tabular}{llll}
($\NOT\NOT_1$) & $\NOT \NOT \varphi / \varphi$ & ($\NOT\NOT_2$) & $\varphi/ \NOT \NOT \varphi$ \\[0,1cm]
($\NOT\AND_1$) & $\NOT (\varphi \AND \psi) / \NOT \varphi \OR \NOT \psi$ & ($\NOT\AND_2$) & $\NOT \varphi \OR \NOT \psi / \NOT (\varphi \AND \psi)$ \\[0,1cm]
($\NOT\OR_1$) & $\NOT (\varphi \OR \psi) / \NOT \varphi \AND \NOT \psi$ & ($\NOT\OR_2$) & $\NOT \varphi \AND \NOT \psi / \NOT (\varphi \OR \psi)$ \\[0,1cm]
($\NOT{\IMPL}_1$) & $\NOT (\varphi \IMPL \psi) / \POSS (\varphi \AND \NOT \psi)$ & ($\NOT{\IMPL}_2$) & $\POSS (\varphi \AND \NOT \psi)/ \NOT (\varphi \IMPL \psi)$ \\[0,1cm]
(CEM) & $/ (\varphi \IMPL \bot) \OR \POSS \varphi$ & ($\POSS\oplus$) & $\POSS \alpha_{1} \AND \ldots \AND \POSS \alpha_{n}/\POSS (\alpha_{1} \oplus \ldots \oplus \alpha_{n})$ \\[0,1cm]
%($\POSS\OR$) & $\POSS (\varphi \OR \psi)/\POSS \varphi \OR \POSS \psi$ & ($\POSS\oplus$) & $\POSS \alpha_{1} \AND \ldots \AND \POSS \alpha_{n}/\POSS (\alpha_{1} \oplus \ldots \oplus \alpha_{n})$ \\
%(CEM) & $/ (\varphi \IMPL \bot) \OR \POSS \varphi$ &  & \\
\end{tabular}

\vspace{6pt}

We will write $\varphi_1, \ldots, \varphi_n \vdash \psi$ if $\psi$ is derivable in this system from the assumptions $\varphi_1, \ldots, \varphi_n$.% Besides the rules (i$\sim$), (i$\NOT$) we will need the following forms of indirect proof.

\vspace{6pt}

As we already mentioned, our semantics has some connection to inquisitive semantics. Note, however, that our deductive system is very different from the standard system of natural deduction for inquisitive logic (see, e.g., \cite{ciardelli16}), though it has some similarities with the system for inquisitive logic with weak negation developed in \cite{puncochar15}. 

We can illustrate the role of restrictions related to round brackets with the following example. If the restriction given by the round brackets were not present we could derive, for example, the contradiction $\bot$ from the premise $\POSS p \AND \POSS {\sim} p$ in the following way:

{\footnotesize \begin{itemize}
\item[]
\begin{fitch}
					\POSS p \AND \POSS {\sim} p = \NOT (p \IMPL \bot) \AND \NOT ({\sim} p \IMPL \bot)							& premise \\ %1
					p \cup {\sim} p	 						& the standard derivation of excluded middle \\			%2
\fh			 		p											& hypothetical assumption \\			%3
\fa \fh		 	{\sim} p									& hypothetical assumption \\	%4
\fa	 \fa			\bot										& 3,4, (e${\sim}_1$) \\	%5
\fa					{\sim} p \IMPL \bot					& 4-5, (i$\IMPL$) \\	%6
\fa					\NOT ({\sim} p \IMPL \bot)     	& 1, (e$\AND_2$)  \\			%7
\fa 				\bot		                 					& 6,7, (e$\NOT$) \\			%8
\fh			  		{\sim} p						             & hypothetical assumption \\	    %9
\fa \fh			p                      						& hypothetical assumption \\%10
\fa \fa			\bot	 									& 9,10, (e${\sim}_1$)  \\	%11
\fa 				p \IMPL \bot							& 10-11, (i$\IMPL$)  \\	%12
\fa					\NOT (p \IMPL \bot) 				& 1, (e$\AND_1$)  \\	%13
\fa 				\bot										& 12,13, (e$\NOT$)  \\	%14
					\bot				& 2,3-8,9-14, (e$\cup$) \\	%15
\end{fitch}
\end{itemize}} 

This would be an undesirable result. In our semantics, $\POSS p \AND \POSS {\sim} p$ does not entail $\bot$ which accords with the intuition that $p$ and ${\sim} p$  can consistently be both possible.  As it stands, the derivation is incorrect in the system, because in the steps 7 and 13 the restriction on the rule  (e$\cup$) was not respected and the unsafe formula $\POSS p \AND \POSS {\sim} p$ occurring in the outer proof was used under the hypothetical assumptions.

We can observe that all rules of the system are sound with respect to the semantics. To illustrate how the soundness proof goes let us consider the rules (e$\cup$) and (e$\OR$). Let $\Delta^s$ denote the set of safe formulas from $\Delta$. Soundness of the two rules corresponds respectively to the following two semantic facts:
\begin{itemize}
\item[(a)] If $\Delta^s, \alpha \vDash \gamma$ and $\Delta^s, \beta \vDash \gamma$ then $\Delta, \alpha \cup \beta \vDash \gamma$.
\item[(b)] If $\Delta, \varphi \vDash \chi$ and $\Delta, \psi \vDash \chi$ then $\Delta, \varphi \OR \psi \vDash \chi$.
\end{itemize}
In order to prove (a), assume that $\Delta^s, \alpha \vDash \gamma$ and $\Delta^s, \beta \vDash \gamma$. Let $C$ be a context in which all formulas from $\Delta$ and the formula $\alpha \cup \beta$  are assertible. Take an arbitrary possible world $w \in C$. All formulas of $\Delta^s$ are assertible in $\{w\}$ (due to persistence of safe formulas). Moreover, $\alpha$ or $\beta$ is assertible in $\{w\}$. It follows from our assumption that $\gamma$ is assertible in $\{w\}$. Since this holds for every $w \in C$, and $\gamma$ is an $L$-formula, we obtain $C \Vdash^+ \gamma$, as required. 

In order to prove (b), assume that $\Delta, \varphi \vDash \chi$ and $\Delta, \psi \vDash \chi$. Let $C$ be a context in which all formulas from $\Delta$ and the formula $\varphi \OR \psi$  are assertible. Then $\varphi$ or $\psi$ is assertible in $C$. It follows from our assumption that $\chi$ is assertible in $C$, as required.

\section{Completeness}

%A set of $L^*$-formulas $\Delta$ is consistent if and only if $\Delta \nvdash \bot$. For a set of atomic formulas $A$, an $A$-world is a function from $A$ to truth values. An $A$-context is a set of $A$-worlds. 

%\begin{itemize}
%\item[1.] Definition of the formula $-\varphi$, for every $L^*$-formula $\varphi$.
%\item[2.] Lemma: $\Delta \nvdash \varphi$ iff $\Delta \cup \{-\varphi\}$ is consistent.
%\item[3.] Definition of the formula $\mu_C$, for every finite set of atomic formulas $A$ and every $A$-context $C$.
%\item[4.] Lemma: Let $\Delta$ be a set of $L^*$-formulas containing only atomic formulas from a finite set $A$. If $\Delta$ is consistent, then there is an $A$-context $C$ such that $\Delta \cup \left\{\mu_{C}\right\}$ is consistent.
%\item[5.] Lemma: Let $C$ be an $A$-context, where $A$ is finite. Then for every $A$-formula $\varphi$, either $\mu_{C} \vdash %\varphi$, or $\mu_{C} \vdash -\varphi$.
%\end{itemize}

The proof of completeness proceeds in the following steps. First a contextual weak negation $-\varphi$  is defined recursively. This negation is a denial of assertibility, that is, $-\varphi$ is assertible in a given context iff $\varphi$  is not assertible in that context. It has to be shown that this negation behaves properly also on the syntactic side. That means that the following holds: $\Delta \nvdash \varphi$ if and only if $\Delta \cup \{-\varphi\}$ is consistent, i.e. $\Delta, -\varphi \nvdash \bot$. The proof of this fact will be the main task of this section. The reconstruction of weak negation allows us to reduce completeness to the claim that every consistent set has a model. We will only sketch the proof of this claim because the technique is basically the same as the one used in the completeness proof for inquisitive logic with weak negation from \cite{puncochar15}. The ``contextual weak negation'' is defined by the following five recursive clauses. The first clause states the definition for $L$-formulas and their intensional negations:

\vspace{6pt}

\begin{tabular}{lll}
1. & $-\alpha=\POSS \NOT \alpha$,  & $-\NOT \alpha = \POSS \alpha$.\\
\end{tabular}

\vspace{6pt}

In the clauses 2.-5. we assume that $\varphi, \psi$ are arbitrary $L^*$-formulas for which $-\varphi, -\NOT \varphi, -\psi, -\NOT\psi$ are already defined, and we further define:

\vspace{6pt}

\begin{tabular}{lll}
2. & $-\NOT \NOT \varphi = - \varphi$.\\
3. & (a) $-(\varphi \IMPL \psi)=\POSS (\varphi \AND -\psi)$, & (b) $-\NOT(\varphi \IMPL \psi)=\varphi \IMPL -\NOT \psi$.\\
4. & (a) $-(\varphi \AND \psi)=-\varphi \OR -\psi$, &  (b) $-\NOT (\varphi \AND \psi)=-\NOT \varphi \AND -\NOT \psi$. \\
5. & (a) $-(\varphi \OR \psi)=-\varphi \AND -\psi$, & (b) $-\NOT (\varphi \OR \psi)=-\NOT \varphi \OR -\NOT \psi$.
\end{tabular}

\vspace{6pt}

Note that by these clauses $-\varphi$ is indeed defined for every $L^*$-formula $\varphi$. By induction on $\varphi$, we obtain the following observation.

\begin{lemma}\label{l: ass of minus is the lack of ass}
For any $L^*$-formula $\varphi$ and any context $C$, it holds that
\begin{itemize}
\item[] $C \Vdash^{+} -\varphi$ iff  $C \nVdash^{+} \varphi$.
\end{itemize}
\end{lemma}
We say that a context $C$ is a model of a set of formulas $\Delta$ if and only if every formula from $\Delta$ is assertible in $C$. The following lemma follows directly from Lemma \ref{l: ass of minus is the lack of ass}.
\begin{lemma}\label{l: entailment and models}
$\Delta \nvDash \varphi$ iff $\Delta \cup \{-\varphi\}$ has a model.
\end{lemma}
The next lemma is a cornerstone of the completeness proof and the main technical issue of this paper.

\begin{lemma}\label{l: excluded middle}
For any formula $\varphi$, the following holds:
\begin{itemize}
\item[(a)] $\vdash \varphi \OR - \varphi$,
\item[(b)] $\varphi, -\varphi \vdash \bot$.
\end{itemize}
\end{lemma}
\begin{proof}
(a) We will proceed by simultaneous induction on $\varphi$ and $\NOT \varphi$. In the derivations below, whenever we use a safe formula in a context in which it is required to use only safe formulas, we indicate this in the corresponding annotation.

1. Assume that $\alpha$ is from $L$. We will derive $\alpha \OR -\alpha$, i.e. $\alpha \OR \NOT (\NOT \alpha \IMPL \bot)$. The derivation of $\NOT \alpha \OR -\NOT \alpha$ is similar.

{\footnotesize \begin{itemize}
\item[]
\begin{fitch}
			(\NOT \alpha \IMPL \bot) \OR \NOT (\NOT \alpha \IMPL \bot) 		& (CEM) \\ %1
\fh			 \NOT \alpha \IMPL \bot 						& hypothetical assumption \\			%2
\fa \fh 	{\sim}\alpha 									& hypothetical assumption\\			%3
\fa \fa \fh \alpha											& hypothetical assumption\\%4
\fa \fa \fa \bot											& 3,4, (e$\sim_1$)\\ %5
\fa \fa \NOT \alpha										& 4-5, (i$\NOT$) \\%6
\fa \fa	\bot											& 2-safe,6, (e$\IMPL$) \\	%7
\fa 		{\sim}{\sim}\alpha								& 3-7, (i$\sim$)  \\			%8
\fa 		\alpha										& 8, (e$\sim_2$)\\ %9
\fa 		\alpha \OR \NOT (\NOT \alpha \IMPL \bot)                                   & 9, (i$\OR_1$) \\			%10
\fh 		\NOT (\NOT \alpha \IMPL \bot)                                                     & hypothetical assumption \\			%11
\fa 		\alpha \OR \NOT (\NOT \alpha \IMPL \bot)                                   & 11, (i$\OR_2$)\\			%12
 		\alpha \OR \NOT (\NOT \alpha \IMPL \bot)					& 1,2-10,11-12, (e$\OR$)  \\	%13

\end{fitch}
\end{itemize}} 
In the cases 2.-5. we assume as an inductive hypothesis that our claim holds for $\varphi, \NOT \varphi, \psi, \NOT \psi$. 

2. It is easy to derive $\NOT \NOT \varphi \OR -\NOT \NOT \varphi$ from $\varphi \OR -\varphi$, by (e$\OR$), (i$\OR_1$), (i$\OR_2$), and ($\NOT\NOT_2$).
\newpage

3. We show that our claim holds for $\varphi \IMPL \psi$ and $\NOT (\varphi \IMPL \psi)$. First, we prove $\vdash (\varphi \IMPL \psi) \OR -(\varphi \IMPL \psi)$, i.e. $\vdash (\varphi \IMPL \psi) \OR \POSS (\varphi \AND -\psi)$, which can be done in the following way:

{\footnotesize\begin{itemize}
\item[]
\begin{fitch}
		((\varphi \AND -\psi) \IMPL \bot ) \OR \POSS (\varphi \AND -\psi)	&  (CEM) \\		
\fh		(\varphi \AND -\psi) \IMPL \bot	& hypothetical assumption \\
\fa \fh	\varphi 						& hypothetical assumption \\
\fa \fa	\psi \OR - \psi				& induction hypothesis \\
\fa \fa \fh	\psi						& hypothetical assumption \\
\fa \fa \fh	-- \psi						& hypothetical assumption		\\			 
\fa \fa \fa 	\varphi \AND - \psi				& 3,6, (i$\AND$)		\\	
\fa \fa \fa 	\bot						& 2-safe,7, (e$\IMPL$)		\\	
\fa \fa \fa 	\psi						& 8, (EFQ) 		\\	
\fa  \fa 	\psi						& 4,5,6-9, (e$\OR$)		\\	
\fa 		\varphi \IMPL \psi				& 3-10, (i$\IMPL$)		\\	
\fa		(\varphi \IMPL \psi) \OR \POSS (\varphi \AND -\psi)	& 11, (i$\OR_1$)	\\	
\fh		 \POSS (\varphi \AND -\psi)					& hypothetical assumption	\\	
\fa 		(\varphi \IMPL \psi) \OR \POSS (\varphi \AND -\psi)	& 13, (i$\OR_2$)	\\	
		(\varphi \IMPL \psi) \OR \POSS (\varphi \AND -\psi)	& 1,2-12,13-14, (e$\OR$)		\\			
\end{fitch}
\end{itemize}}

Now we prove $\vdash \NOT (\varphi \IMPL \psi) \OR - \NOT (\varphi \IMPL \psi)$, i.e. $\NOT (\varphi \IMPL \psi) \OR  (\varphi \IMPL - \NOT \psi)$:

{\footnotesize \begin{itemize}
\item[]
\begin{fitch}
		((\varphi \AND \NOT \psi) \IMPL \bot ) \OR \POSS (\varphi \AND \NOT \psi)	&  (CEM) \\		
\fh		(\varphi \AND \NOT \psi) \IMPL \bot	& hypothetical assumption \\
\fa \fh	\varphi							& hypothetical assumption \\
\fa \fa	\NOT \psi \OR -\NOT \psi			& induction hypothesis \\
\fa \fa \fh	\NOT \psi						& hypothetical assumption \\		 
\fa \fa \fa 	\varphi \AND \NOT \psi				& 3,5, (i$\AND$)		\\	
\fa \fa \fa 	\bot							& 2-safe,6, (e$\IMPL$)		\\	
\fa \fa \fa 	-- \NOT \psi						& 7, (EFQ)	\\	
\fa  \fa \fh 	-- \NOT \psi						& hypothetical assumption		\\
\fa \fa	-- \NOT \psi						& 4,5-8,9, (e$\OR$)		\\	
\fa		\varphi \IMPL -\NOT \psi			& 3-10, (i$\IMPL$)	\\	
\fa		\NOT (\varphi \IMPL \psi) \OR  (\varphi \IMPL - \NOT \psi)	& 11, (i$\OR_2$) 		\\	
\fh		 \POSS (\varphi \AND \NOT \psi)						& hypothetical assumption	\\	
\fa 		\NOT (\varphi \IMPL \psi)							& 13, ($\NOT{\IMPL}_2$)	\\	
\fa		\NOT (\varphi \IMPL \psi) \OR  (\varphi \IMPL - \NOT \psi)	& 14, (i$\OR_1$)	\\			
		\NOT (\varphi \IMPL \psi) \OR  (\varphi \IMPL - \NOT \psi)	& 1,2-12,13-15, (e$\OR$)		\\
\end{fitch}
\end{itemize}} 

\newpage

4. We prove that our claim holds for $\varphi \AND \psi$ and $\NOT (\varphi \AND \psi)$. First, we prove the former, i.e. $\vdash (\varphi \AND \psi) \OR (- \varphi \OR - \psi)$, which can be done by the following derivation:

{\footnotesize \begin{itemize}
\item[]
\begin{fitch}
					\varphi \OR - \varphi							& induction  hypothesis \\ %1
\fh			 		\varphi		 										& hypothetical assumption	 \\			%2
\fa		 			\psi \OR - \psi										& induction  hypothesis \\	%3
\fa \fh			\psi													& hypothetical assumption	 \\	%4
\fa \fa			\varphi \AND \psi								& 2,4, (i$\AND$)  \\	%5
\fa \fa 			(\varphi \AND \psi) \OR (-\varphi \OR -\psi)		                 & 5, (i$\OR_1$) \\			%6
\fa \fh	  		-\psi							                                                     & hypothetical assumption	 \\	    %7
\fa \fa		 	(\varphi \AND \psi) \OR (-\varphi \OR -\psi)	                       & 7, (i$\OR_2$) (twice) \\%8
\fa				 	(\varphi \AND \psi) \OR (-\varphi \OR -\psi)	 					& 3, 4-6,7-8, (e$\OR$)  \\	%9
\fh				 	- \varphi																	& hypothetical assumption	   \\%10
\fa					(\varphi \AND \psi) \OR (-\varphi \OR -\psi)			    & 10, (i$\OR_1$), (i$\OR_2$) \\	%11
					(\varphi \AND \psi) \OR (-\varphi \OR -\psi)  		         & 1,2-9,10-11, (e$\OR$) \\			%12
\end{fitch}
\end{itemize}}

Now we prove that $\vdash \NOT (\varphi \AND \psi) \OR (-\NOT \varphi \AND -\NOT \psi)$, which can be done by the following derivation:

{\footnotesize \begin{itemize}
\item[]
\begin{fitch}
					\NOT \varphi \OR - \NOT \varphi							& induction  hypothesis \\ %1
\fh			 		\NOT \varphi		 												& hypothetical assumption	 \\			%2
\fa		 			\NOT \varphi \OR \NOT \psi									& 2, (i$\OR_1$) \\	%3
\fa					\NOT (\varphi \AND \psi)										& 3, ($\NOT\AND_2$) \\	%4
\fa					\NOT (\varphi \AND \psi) \OR (-\NOT \varphi \AND -\NOT \psi)		& 4, (i$\OR_1$)  \\	%5
\fh		 			-\NOT \varphi		                 & hypothetical assumption	 \\			%6
\fa			  		\NOT \psi \OR - \NOT \psi							                                        & induction  hypothesis \\	    %7
\fa \fh		 	\NOT \psi                       & hypothetical assumption	 \\%8
\fa \fa 			\NOT \varphi \OR \NOT \psi									& 8, (i$\OR_2$) \\	%9
\fa \fa			\NOT (\varphi \AND \psi)													& 9, ($\NOT\AND_2$)  \\	%10
\fa \fa			\NOT (\varphi \AND \psi) \OR (-\NOT \varphi \AND -\NOT \psi) & 10, (i$\OR_1$) \\	%11
\fa \fh	  		- \NOT \psi							                                                     & hypothetical assumption	 \\ %12
\fa	 \fa		 	 (-\NOT \varphi \AND -\NOT \psi)	 					& 6,12, (i$\AND$)  \\	%13
\fa \fa		 	\NOT (\varphi \AND \psi) \OR (-\NOT \varphi \AND -\NOT \psi)		& 13, (i$\OR_2$)   \\%14
\fa					\NOT (\varphi \AND \psi) \OR (-\NOT \varphi \AND -\NOT \psi) & 7,8-11,12-14, (e$\OR$)  \\	%15
					\NOT (\varphi \AND \psi) \OR (-\NOT \varphi \AND -\NOT \psi)  & 1,2-5,6-15, (e$\OR$) \\			%16
\end{fitch}
\end{itemize}} 

\newpage
5. The case of $\OR$ is analogous to the case of $\AND$. This finishes the proof of (a).

(b) We will proceed again by simultaneous induction on $\varphi$ and $\NOT \varphi$. 1. Assume that $\alpha$ is any $L$-formula. First, we will derive $\bot$ from $\alpha$ and $-\alpha = \NOT (\NOT \alpha \IMPL \bot)$. The derivation of $\bot$ from $\NOT \alpha$ and $-\NOT \alpha$ is similar.

{\footnotesize \begin{itemize}
\item[]
\begin{fitch}
			\alpha 																					& premise \\ %1
			\NOT (\NOT \alpha \IMPL \bot) 												& premise \\	%2
\fh		 	\NOT \alpha 																			& hypothetical assumption  \\	%3
\fa			\bot																						& 1-safe,3, (e$\NOT$) \\	%4
	 		\NOT \alpha \IMPL \bot															& 2-3, (i$\IMPL$)  \\			%5
			\bot                                  & 2,5, (e$\NOT$) \\			%6
\end{fitch}
\end{itemize}} 
For the steps 3.-5., assume as an induction hypothesis that our claim holds for some arbitrary $\varphi, \NOT \varphi, \psi, \NOT \psi$. 2. It is easy to see that if we assume $\varphi, -\varphi \vdash \bot$, then also $\NOT \NOT \varphi, -\NOT \NOT \varphi \vdash \bot$, by ($\NOT\NOT_1$).  

3. We prove that the claim holds for $\varphi \IMPL \psi$, i.e. $\varphi \IMPL \psi, -(\varphi \IMPL \psi) \vdash \bot$. 

{\footnotesize \begin{itemize}
\item[]
\begin{fitch}
			\varphi \IMPL \psi 		& premise \\
			--(\varphi \IMPL \psi)=\NOT ((\varphi \AND - \psi ) \IMPL \bot) & premise \\
\fh			\varphi \AND - \psi 	& hypothetical assumption \\
\fa			\varphi 				& 3, (e$\AND_1$) \\
\fa			\psi 				& 1-safe,4, (e$\IMPL$) \\
\fa			-- \psi 				& 3, (e$\AND_2$) \\
\fa			\bot 				& 5,6, induction hypothesis\\
			(\varphi \AND - \psi) \IMPL \bot & 3-7, (i$\IMPL$) \\
			\bot				& 2,8, (e$\NOT$) \\
\end{fitch}
\end{itemize}}

Now we prove that $\NOT (\varphi \IMPL \psi), -\NOT (\varphi \IMPL \psi) \vdash \bot$. 

{\footnotesize \begin{itemize}
\item[]
\begin{fitch}
			\NOT (\varphi \IMPL \psi)							& premise \\
			-- \NOT (\varphi \IMPL \psi)= \varphi \IMPL - \NOT \psi 		& premise \\
\fh			\varphi \AND \NOT \psi 								& hypothetical assumption \\
\fa			\varphi 											& 3, (e$\AND_1$) \\
\fa			-- \NOT \psi 										& 2-safe,4, (e$\IMPL$) \\
\fa			\NOT \psi										& 3, (e$\AND_2$) \\
\fa			\bot 											& 5,6,   induction hypothesis\\
			(\varphi\AND \NOT \psi) \IMPL \bot 					& 3-7, (i$\IMPL$) \\
			\NOT ((\varphi \AND \NOT \psi) \IMPL \bot) 				& 1, ($\NOT{\IMPL}_1$) \\
			\bot											& 8,9, (e$\NOT$) \\
\end{fitch}
\end{itemize}}
\newpage
4. We prove that our claim holds for $\varphi \AND \psi$ and $\NOT (\varphi \AND \psi)$. First, we prove that $\varphi \AND \psi, -(\varphi \AND \psi) \vdash \bot$, which can be done by the following derivation:

{\footnotesize \begin{itemize}
\item[]
\begin{fitch}
					\varphi \AND \psi						& premise \\ %1
			 		-\varphi \OR -\psi		 				& premise \\			%2
\fh		 			-\varphi										& hypothetical assumption\\	%3
\fa					\varphi										& 1, (e$\AND_1$) \\	%4
\fa					\bot											& 3,4, induction hypothesis  \\	%5
\fh 				-\psi		                 						& hypothetical assumption \\			%6
\fa			  		\psi							                  & 1, (e$\AND_1$) \\	    %7
\fa				 	\bot                       						& 6,7, induction hypothesis \\%8
				 	\bot	 										& 2,3-5,6-8, (e$\OR$)  \\	%9
\end{fitch}
\end{itemize}} 
Now we prove that $\NOT(\varphi \AND \psi), -\NOT (\varphi \AND \psi) \vdash \bot$, which can be done by the following derivation:

{\footnotesize \begin{itemize}
\item[]
\begin{fitch}
					\NOT (\varphi \AND \psi)								& premise \\ %1
			 		-\NOT \varphi \AND -\NOT \psi		 				& premise \\			%2
			 		\NOT \varphi \OR \NOT \psi							& 1, ($\NOT\AND_1$) \\			%3
\fh		 			\NOT \varphi												& hypothetical assumption \\	%4
\fa					-\NOT \varphi												& 2, (e$\AND_1$) \\	%5
\fa					\bot															& 4,5, induction hypothesis  \\	%6
\fh 				\NOT \psi		                 								& hyp \\			%7
\fa			  		-\NOT \psi							                 			& 2, (e$\AND_2$) \\	    %8
\fa				 	\bot                       										& 7,8, induction hypothesis \\%9
				 	\bot	 														& 3,4-6,7-9, (e$\OR$)  \\	%10
\end{fitch}
\end{itemize}} 
5. The case of $\OR$ is analogous to the case of $\AND$. This finishes the proof of (b).
\end{proof}

\begin{lemma}\label{l: defined negation and consistency}
$\Delta \nvdash \varphi$ iff $\Delta \cup \{-\varphi\}$ is consistent.
\end{lemma}
\begin{proof}
The left-to-right implication is obtained from Lemma \ref{l: excluded middle}-a, using the rules (e$\OR$) and (EFQ), and the right-to-left implication follows from Lemma \ref{l: excluded middle}-b.
\end{proof}

\begin{lemma}\label{l: consistent sets have model}
Every consistent finite set of $L^*$-formulas has a model.
\end{lemma}
\begin{proof}
We can give only a sketch of the proof here. The strategy is the same as in the analogous proofs in \cite{puncochar15} and \cite{puncochargauker20}. Let $A$ be a finite set of atomic formulas. Recall that for every $A$-context $C$ there is an $L^*$-formula $\mu_C$ characterizing $C$ in the following sense: For every $A$-context $D$, $D \Vdash^{+} \mu_{C}$ iff $D=C$. We can further prove that if $\Delta$ is a consistent set of $L^*$-formulas built out of the atomic formulas from $A$, then there is an $A$-context $C$ such that $\Delta \cup \left\{\mu_{C}\right\}$ is consistent. Moreover, we can prove by induction, crucially using the rule ($\POSS\oplus$), that for every $L^*$-formula $\varphi$ built out of the atomic formulas from $A$, either $\mu_{C} \vdash \varphi$, or $\mu_{C} \vdash -\varphi$. 

We can further reason as follows. Let $\Delta$  be a consistent finite set of $L^*$-formulas and let $A$  be the set of atomic formulas occurring in $\Delta$. Then there is an $A$-context $C$ such that $\Delta \cup \left\{\mu_{C}\right\}$ is consistent. It follows that $\mu_{C} \vdash \psi$, for every $\psi \in \Delta$. Since $C  \Vdash^{+} \mu_{C}$ it follows from soundness of the deductive rules that $C  \Vdash^{+} \psi$, for every $\psi \in \Delta$. Hence, $\Delta$ has a model.
\end{proof}

\begin{theorem}
$\varphi_1, \ldots, \varphi_n \vdash \psi$ \, iff\, $\varphi_1, \ldots, \varphi_n \vDash \psi$.
\end{theorem}
\begin{proof}
Soundness was already discussed and completeness is obtained from Lemmas \ref{l: entailment and models}, \ref{l: defined negation and consistency}, and \ref{l: consistent sets have model}.
\end{proof}

\bibliographystyle{eptcs}
\bibliography{biblio}

\end{document}